\documentclass{sig-alternate-10pt}
\usepackage[latin1]{inputenc} % required for "see section" character
\usepackage{cite}
\usepackage{amsfonts}
\usepackage{xspace}
\usepackage{pifont}

\newtheorem{theorem}{Theorem}

\usepackage{wrapfig}
\usepackage{graphicx}
\usepackage[pdftex]{color} % Force PDF output

\usepackage{fancybox}
\usepackage{calc}
\usepackage{url}
\usepackage[hang,scriptsize,tight,nooneline]{subfigure}


\newcommand{\cut}[1]{}
\newcommand{\fixme}[1]{{\bf\textcolor{red}{[#1]}}}
\newcommand{\cutatlastminute}[1]{}

\newcommand{\parheading}[1]{\medskip{} \noindent \textbf{#1}}

\newcommand{\disco}{Disco\xspace}
\newcommand{\nddcr}{NDDisco\xspace}

\newcommand{\ankitr}{}
\begin{document}

%%%%%%%%%%%%%%%%%%%%%%%%%%%%%%%%%%%%%%%%%%%%%%%%%%%%%%%%%%%%%%%%%%%%%%%%%%%%%%%%

\title{Scalable Routing on Flat Names}
%{\large \fixme{Draft --- Please do not distribute}}}
%{\large Paper \#1569334709, 12 pages}}

\numberofauthors{2}
\author{
\alignauthor
Ankit Singla and P. Brighten Godfrey\\
       \affaddr{University of Illinois at Urbana-Champaign}\\
%       \email{\{singla2,pbg\}@illinois.edu}
\alignauthor
Kevin Fall, Gianluca Iannaccone,\\ and Sylvia Ratnasamy\\
       \affaddr{Intel Labs Berkeley}\\
%\email{\{kevin.fall,gianluca.iannaccone,sylvia.p.ratnasamy\}@intel.com}
}

%\author{Ankit Singla\uiuc, P. Brighten Godfrey\uiuc, Kevin Fall, Gianluca Iannaccone, Sylvia Ratnasamy\ilb\\
%\affaddr{\uiuc University of Illinois at Urbana-Champaign~~~~~~~~~~\ilb UC Berkeley\\
%\email{pbg@illinois.edu, \{igor,shenker,istoica\}@cs.berkeley.edu}}}

\date{~}

\maketitle
\vspace{-0.80in}

%%%%%%%%%%%%%%%%%%%%%%%%%%%%%%%%%%%%%%%%%%%%%%%%%%%%%%%%%%%%%%%%%%%%%%%%%%%%%%%%

\begin{abstract}

\cut{Common wisdom is that networks scale to large size through hierarchy:
routing is performed on higher level aggregates and subsequently at
finer granularity.  The consequences are inflated path lengths, and the
use of location-dependent addresses which complicate management,
mobility, and multihoming.}
\cut{This paper takes a different approach.}

\textnormal{We introduce a protocol\let\thefootnote\relax\footnotetext{This technical 
report extends our ACM CoNEXT 2010 paper~\cite{singla10disco} by including the proofs for the theoretical results.}
which routes on flat, location-independent identifiers with guaranteed
scalability and low stretch. Our design builds on theoretical advances in
the area of compact routing, and is the first to realize these
guarantees in a dynamic distributed setting.}

\end{abstract}

%%%%%%%%%%%%%%%%%%%%%%%%%%%%%%%%%%%%%%%%%%%%%%%%%%%%%%%%%%%%%%%%%%%%%%%%%%%%%%%%

\section{Introduction}

Routing scalability is highly desirable for very large, dynamic, and
resource-constrained networks, including many peer-to-peer systems and
the Internet.  Shortest-path routing algorithms (link state, distance
vector, path vector, etc.) all~\cite{fg95} require $\Omega(n)$ memory at
each router for a network with $n$ destinations, and at least as much
communication and computation to build the routing tables.

One way to scale routing is to use a structured network topology that
makes routing easy. Examples range from torus networks in
supercomputers~\cite{kessler1993cray} to planar networks which permit
greedy geographic routing~\cite{karp2000gpsr} to hypercubes, small world
networks, and other topologies in distributed hash
tables~\cite{pastry,zhao01tapestry,stoica01chord}.  But requiring a
particular highly structured topology is not feasible for
general-purpose networks.

\cut{But
requiring a particular highly structured topology is not feasible for
general-purpose networks.  Even for some P2P networks, assuming a particular
topology is infeasible: the Unmanaged Internet
Architecture~\cite{ford2006persistent}, for example, maintains links
only between friends' devices.}

For general networks, the common scaling technique is {\em hierarchy}: 
routing is performed over high-level aggregate units until it reaches
the destination's unit, at which point routing proceeds at a finer
granularity.  For example, the Internet routes at the level of IP
prefixes to a destination domain, and then over an intradomain routing
protocol to a subnet\cut{, and then often via Ethernet's spanning tree
protocol to the destination}. Hierarchy has two main problems.  First, it
can have arbitrarily high {\bf stretch}, defined as the ratio of route
length to the shortest path length.  Simultaneously guaranteeing
scalability and low stretch on arbitrary networks is a nontrivial
problem which no deployed routing protocols achieve.  Second, hierarchy
requires {\bf location-dependent addresses}, complicating management,
mobility, and multihoming.

In response, numerous recent proposals suggest routing on
location-independent {\bf flat names} \cite{gc01, stoica2002internet,
cbfp03, jokela2004host, lisp, koponen2007data,
andersen2008accountable}.  Flat names are a paradigm shift for the
network layer: rather than requiring a location-dependent IP address to
serve the needs of the routing protocol, a name is an arbitrary bit
string that can serve the needs of the application layer.  For example,
a name could be a DNS name, a MAC address, or a secure self-certifying
identifier~\cite{koponen2007data,stoica2002internet,jokela2004host}. But
no previously proposed protocols exist for scalable, low-stretch routing
on flat names.

This paper introduces a new routing protocol: Distributed Compact
Routing, or ``\disco''.  \disco builds on theoretical advances in
centralized, static compact routing algorithms~\cite{tz01,agmnt04} and
is the first {\em dynamic distributed} protocol to guarantee the following
properties:

\begin{description}

\item[Scalability:] A \disco router needs just $\tilde{O}(\sqrt{n})$
	routing table entries\footnote{The standard notation
	$\tilde{O}(\cdot)$ hides polylogarithmic factors to aid
	readability. \disco actually has $O(\sqrt{n \log n})$ entries.} regardless of network topology.
	
	\cut{or $\tilde{O}(\sqrt{n})$ {\em bits} of state under an assumption bounding
	the size of a partial source route, which we argue
	is a reasonable assumption in practice (\S\ref{sec:snip}).}

\item[Low stretch:] \disco has worst-case stretch $7$ on a
	flow's first packet, worst-case stretch $3$ on subsequent
	packets, and much lower average stretch.
\item[Flat names:] \disco routes on arbitrary ``flat'' names
	rather than hierarchical addresses.

\end{description}

\disco's guarantee of $\tilde{O}(\sqrt{n})$ routing table {\em entries}
translates to $\tilde{O}(r\sqrt{n})$ {\em bits} of state where $r$ is the size of a partial route; see discussion in \S\ref{sec:requirements}.

Recently, several papers have introduced techniques which approach the
above properties.  In particular, citing the lack of a distributed
compact routing protocol, VRR~\cite{caesar2006virtual} introduced a
novel DHT-inspired approach to route on flat names, but as we will see,
VRR does not guarantee scalability and low stretch.  S4~\cite{mwqls07} is a
distributed implementation of a compact routing algorithm
of~\cite{tz01}, but it does not route on flat names and as we will show,
it breaks the state bound of~\cite{tz01} causing high per-node state on
realistic topologies.  Distributed compact routing was listed as an open
problem by Gavoille~\cite{g06}.

We find, however, that guaranteed scalable, efficient routing on flat names is achievable. \disco is a careful synthesis of compact routing theory with well-understood systems techniques and a novel, low-overhead overlay network for disseminating routing state. That synthesis is the primary contribution of this paper. 
\disco thus represents a step towards closing the gap
between theoretical and applied work on scalable routing.

\cut{Nevertheless,
much work remains, including application to specific domains.  To spur further work
in this direction, we include a discussion of one application domain:
support for policy in interdomain environments.}

The rest of this paper proceeds as follows. Sec.~\ref{sec:requirements}
discusses the need for our three requirements of scalability,
low stretch, and flat names.  We describe how past work has fallen short
of these requirements in Sec.~\ref{sec:related}. 
Sec.~\ref{sec:protocol} presents our protocol and Sec.~\ref{sec:eval}
evaluates it. \cut{In Sec.~\ref{sec:policy}, we explore how \disco can
support operator control of routing policy.} We conclude with a discussion of future work in
Sec.~\ref{sec:conclusion}.

%%%%%%%%%%%%%%%%%%%%%%%%%%%%%%%%%%%%%%%%%%%%%%%%%%%%%%%%%%%%%%%%%%%%%%%%%%%%%%%%

\section{Requirements}
\label{sec:requirements}

This paper is guided by three key requirements.

\parheading{Guaranteed scalability:} The routing protocol should use little space and messaging regardless of the network topology. In particular, we wish to reduce state requirements from the $\Omega(n)$ bits at each node used by traditional routing protocols in an $n$-node network, to something asymptotically smaller, e.g., $\tilde{O}(\sqrt{n})$ per node. Guaranteeing scalability on any topology helps ensure that the protocol will continue to function smoothly despite future growth, higher dynamics, and new environments.

\disco meets this requirement for graphs in which we have a bound on the size of ``explicit routes'' within the local vicinity of each node.  In particular, \disco has $\tilde{O}(\sqrt{n})$ routing table entries for a total of $\tilde{O}(r\sqrt{n})$ bits of state where $r$ is the maximum size of an explicit route, since these routes are embedded in nodes' addresses. (For example, we might have $r=O(\log n)$.) In \S\ref{sec:snip}, we discuss the need for this assumption and argue that it is reasonable. While in the worst case $r = \tilde{O}(\sqrt{n})$, in a router-level Internet map our addresses are shorter than IPv6 addresses.

\parheading{Guaranteed low stretch: } The protocol should find paths
that are close to the shortest possible.  Stretch is the ratio of the
protocol's route length to the shortest path length.  Since it is not
possible to route on shortest paths using $o(n)$ state per
node~\cite{fg95}, some stretch is unavoidable given our scaling goal. 
However, we desire stretch to be bounded by a small constant in the
worst case, so that local communication stays local and performance does
not degrade regardless of the traffic demands.

\parheading{Flat names: } The protocol should route on arbitrary node
names which may have no relationship with a node's location. In
particular, this service should operate while preserving the stretch
guarantee. \cut{Flat names stand in contrast to traditional {\em
hierarchical names} like IP addresses, which aid the routing system by
encoding location.}

Flat names are widely recognized as a useful primitive. The
location-independence of flat names aids mobility and eliminates the
management burden of location-based address assignment. Numerous
proposed redesigns of Internet routing use flat names to cleanly
separate location from identity, including TRIAD~\cite{gc01},
i3~\cite{stoica2002internet}, FARA~\cite{cbfp03},
HIP~\cite{jokela2004host}, and LISP~\cite{lisp}. Flat names can also be
self-certifying~\cite{mazieres1999separating,walfish2004untangling,koponen2007data},
where the name is a public key or a hash of a public key.  This
provides security without a public key infrastructure to associate keys
with objects.  Self-certifying names have been proposed for data objects
in content-centric networks~\cite{koponen2007data,sarela2008rtfm}, for a
persistent and contention-free replacement for
URLs~\cite{walfish2004untangling}, and for nodes in an accountable
Internet architecture~\cite{andersen2008accountable}.

If low stretch were not a goal, supporting flat names would be
straightforward. In particular, the routing system could resolve names
into addresses---using DNS, a consistent hashing database over a set of
well-known nodes~\cite{bvr,mwqls07,kcr08}, or a
DHT~\cite{ford08phd}---and then route over addresses.  But these
solutions violate our stretch requirement: the resolution step might
travel across the world even if the destination is next door.\footnote{Locality-aware
DHTs rely on existence of underlying routing infrastructure
and hence don't solve this problem. We discuss this issue in \S\ref{sec:related}.} Even
though the resolution step may  be needed for only the first packet of a
flow, this latency may dominate for short flows; indeed, sub-second
delays make a difference to users in web services such as
Google~\cite{brutlag09}.  Moreover, these solutions lack fate
sharing~\cite{clark1988design}: a failure far from the
source-destination path can disrupt communication.  Similarly, an
attacker far from the path could disrupt, redirect, or eavesdrop on
communication.

\medskip{}
Despite its desirability, scalable low-stretch routing on flat names has
remained elusive.  Indeed, satisfying all three requirements is an
algorithmically challenging goal which remained essentially unsolved
from the 1989 introduction of the problem~\cite{ablp89} until
2003~\cite{aclrt03}, even in the static centralized case.  And as we
shall see in the next section, no {\em distributed} solution has been
developed.

%%%%%%%%%%%%%%%%%%%%%%%%%%%%%%%%%%%%%%%%%%%%%%%%%%%%%%%%%%%%%%%%%%%%%%%%%%%%%%%%

\section{Related Work}
\label{sec:related}

Scalable routing has a long and broad history, beginning with Kleinrock
and Kamoun's 1977 analysis of hierarchical
routing~\cite{kleinrock1977hierarchical} and branching into many
practical and theoretical application domains.  We limit our discussion
to routing protocols intended to scale for arbitrary topologies.

\parheading{Scalable routing in theory.} Universal compact routing schemes, beginning with early work by Peleg and Upfal~\cite{pu88}, bound the size of routing tables and the worst-case stretch on arbitrary undirected graphs. A series of results improved the state and stretch bounds, culminating in Thorup and Zwick~\cite{tz01} who obtained stretch $3$ and space $\tilde{O}(\sqrt{n})$ per node, along with a family of schemes that have stretch $k$ and space $\tilde{O}(k n^{2/(k+1)})$ for any odd integer $k \geq 1$. This is essentially optimal, since stretch $<3$ requires state $\Omega(n)$~\cite{gg01}, stretch $<5$ requires state $\Omega(\sqrt n)$, and in general, subject to a conjecture of Erd\H{o}s, stretch $<k$ requires state $\Omega(n^{2/(k-1)})$~\cite{tz01}.

	\cutatlastminute{
 These results are depicted in Fig.~\ref{fig:tradeoffs}.
	
	\begin{figure}
	\begin{centering}
	\includegraphics[scale=0.64]{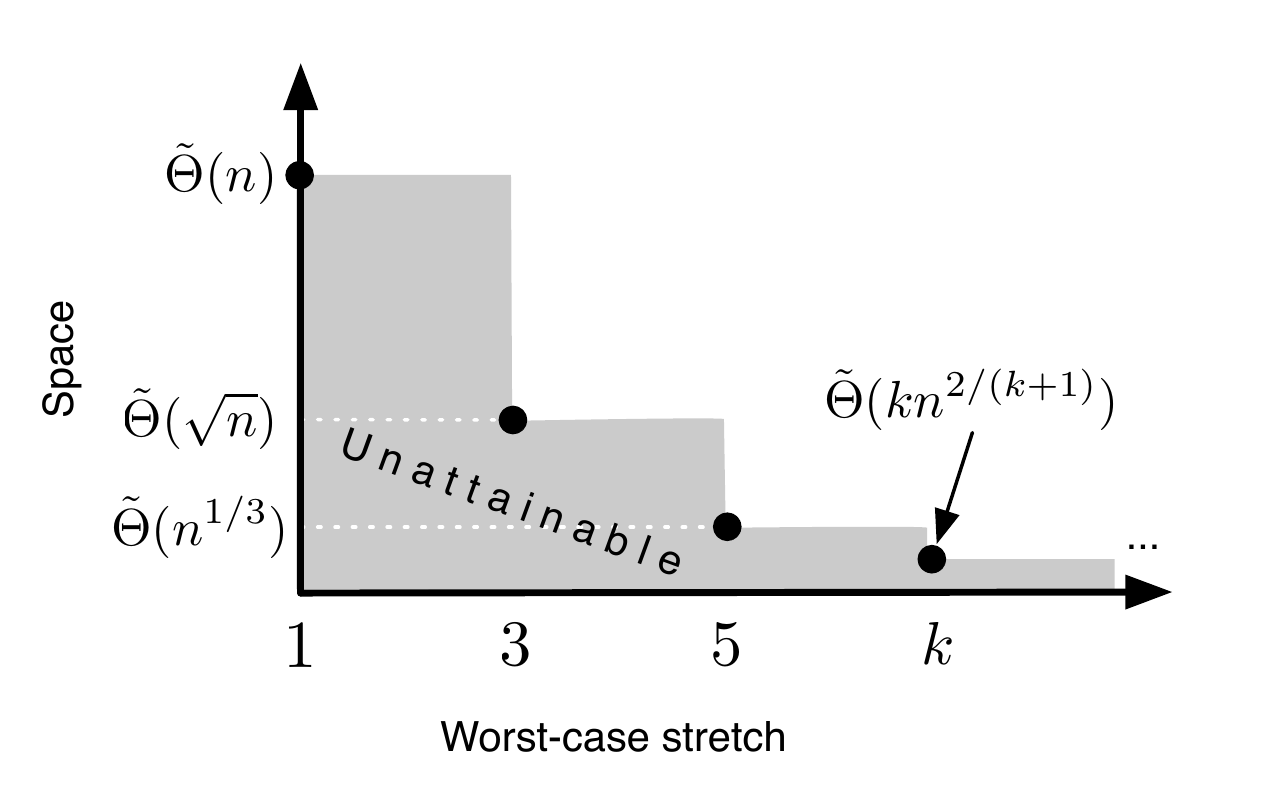}
	\par\end{centering}
	
	\caption{\label{fig:tradeoffs} Scaling in theory.  Shown are
		the best known compact routing schemes
		(in black) and lower bounds (in grey), in the name-dependent model.
		Name-independent results (discussed in the text) are similar.}
	\end{figure}
	
	}
	
Those optimal results, however, are for the {\bf name-dependent} model,
where node names are chosen by the routing algorithm to encode location
information---effectively, they are addresses. In the {\bf
name-independent} model, node names can be arbitrary (``flat'') names.
Abraham et al.~\cite{agmnt04} obtained stretch $3$ with space
$\tilde{O}(\sqrt{n})$, and Abraham, Gavoille, and
Malkhi~\cite{abraham2004routing} obtained stretch $O(k)$ with space
$\tilde{O}(n^{1/k}\log D)$ where $D$ is the normalized diameter of the
network.  Thus, surprisingly, the best name-independent schemes nearly
match the name-dependent ones.

Unfortunately, this theoretical work on compact routing is far from
practical for today's Internet.  Most critically, it assumes centralized
routing table construction and a static network.

	\begin{figure}
	\begin{centering}

{\small
\begin{minipage}{1.0\linewidth}
\centering
\begin{tabular}{l|c|c|c}
%{\bf Scheme} & {\bf Scalable} & {\bf Low stretch} & {\bf Efficient on} \tabularnewline
             &                & {\bf Low }    & {\bf Flat} \tabularnewline
{\bf Scheme} & {\bf Scalable} & {\bf stretch} & {\bf names} \tabularnewline
\hline
\hline 
Shortest-path routing & &  \ding{51} &  \ding{51} \tabularnewline
%\hline 
XL~\cite{levchenko2008xl}  & &  \ding{51} &  \ding{51} \tabularnewline
%\hline 
Classic hierarchy & \ding{51} &  & \tabularnewline
%\hline 
Landmark routing~\cite{t88} & &  & \tabularnewline
%\hline 
BVR~\cite{bvr} & \ding{51} &  & \tabularnewline
%\hline
VRR~\cite{caesar2006virtual} &  &  & \tabularnewline
%\hline
SEATTLE~\cite{kcr08}  & &  \ding{51} &  \tabularnewline
%\hline 
S4~\cite{mwqls07} &  & \ding{51} & \tabularnewline
%\hline
Ford~\cite{ford08phd} & \ding{51} $^a$ & \ding{51} & \tabularnewline
%\hline
%{\em NDDisco} & \ding{51} & \ding{51} & \tabularnewline
%\hline 
{\em \disco} &  \ding{51} \footnote{for graphs where the size of explicit routes is bounded (\S\ref{sec:requirements})} &  \ding{51} &  \ding{51}\tabularnewline
\hline
\end{tabular}
\end{minipage}
	}

	\par\end{centering}
	
	\caption{\label{fig:comparison} Distributed routing protocols.
		{\em Scalable} indicates the protocol guarantees $o(n)$
		state per node; {\em low stretch} indicates it has $O(1)$
		or $O(\log n)$ stretch; {\em flat names} indicates it routes
		on flat names with low stretch. \cut{Note that \disco's (and \cite{ford08phd}'s)
		state guarantee assumes a bound on explicit route size (\S\ref{sec:requirements}).}}
	\end{figure}

\parheading{Scalable routing in practice.}  Since compact
routing algorithms are centralized and often relatively complex, it has
not been clear how to translate them into practical distributed
protocols; despite several attempts, results on the systems side
are limited.  Fig.~\ref{fig:comparison} summarizes some of the protocols
most closely related to \disco.

Landmark routing~\cite{t88} has a similar motivation to our work: small
space and stretch with a dynamic distributed protocol.  However,
landmark routing does not provide guarantees on either space or stretch
for general topologies, and does not route on flat names.

SEATTLE~\cite{kcr08} provides an Ethernet-like protocol while eliminating Ethernet's use of
broadcast.  It looks up Ethernet addresses in a consistent hashing
database run on the routers, and routes along shortest paths after the
first packet.  It therefore does not route on flat names with low
stretch.  SEATTLE improves scalability relative to Ethernet, but still
scales linearly in the number of routers $n$ (and resorts to an
Internet-like hierarchy for large $n$, which would cause unbounded
stretch). \cut{In comparison, \disco provides asymptotically better
scalability without resorting to hierarchy, and bounds stretch even for
the first packet.   As discussed in Sec.~\ref{sec:requirements}, this
improves latency and (via fate-sharing) reliability and security.}

XL~\cite{levchenko2008xl} uses heuristics to significantly improve link
state routing's messaging scalability, but does not improve the worst
case, or routing table size.

Virtual Ring Routing (VRR)~\cite{caesar2006virtual} and Routing on Flat
Labels (ROFL)~\cite{rofl} route on flat names by applying techniques
from DHTs to a physical network rather than in an overlay.  However,
these schemes have unbounded stretch on general topologies, and high
stretch in practice; e.g., an average stretch of up to $8$ in realistic
topologies~\cite{rofl}. We confirm this in §\ref{sec:eval} and also
show that VRR can have very high state for some nodes, even though it
has low state on average.

Beacon Vector Routing (BVR)~\cite{bvr} represents a node address as a
vector of distances to beacons and routes greedily with these vectors. 
While it greatly improves scalability, BVR can get stuck in local minima
causing high stretch, and requires a name resolution step handled by the
landmarks.

S4~\cite{mwqls07} adapts one of the Thorup and Zwick~\cite{tz01} compact
routing schemes to a distributed wireless setting.  However, their
adaptation breaks the per-node state bound of~\cite{tz01}; we show in
§\ref{sec:eval} that S4 can indeed have high per-node state.  Moreover,
it does not route on flat names with low stretch.

Ford~\cite{ford08phd} evaluated a distributed version of~\cite{tz01}
with $\Theta(\log n)$ stretch and state, but~\cite{ford08phd} does not
route on flat names with low stretch. Also, \disco chooses a different
point in the tradeoff space, with $O(1)$ stretch but
$\tilde{O}(\sqrt{n})$ state. Both \disco and~\cite{ford08phd} assume a bound on the size of
an explicit route (\S\ref{sec:requirements}).

	 \cut{ To the best of our
	knowledge, this is the only proposed distributed protocol which is both
	scalable and has low stretch, in this case with $O(\log n)$ state and
	stretch.}

Westphal and Kempf~\cite{westphal2008compact} applied static compact
routing algorithms to support mobility, but only with a fixed
infrastructure and mobile leaf nodes.

Krioukov et al.~\cite{krioukov2004compact,krioukov2007compact} applied
static compact routing algorithms to Internet-like topologies with
promising results, but did not develop distributed protocols.

\parheading{Related techniques in overlay networks.} Distributed hash
tables~\cite{pastry,zhao01tapestry,stoica01chord} may appear to satisfy
our goals, as they route on flat names often with $O(\log n)$ routing
table entries and within $O(\log n)$ overlay hops.  However, DHTs have
unbounded stretch: even a single overlay hop can be arbitrarily longer
than the distance to the destination!  More fundamentally, DHTs are not
general-purpose routing protocols: they are overlay networks which
depend on the availability of a general-purpose routing protocol beneath
them.

%are overlay networks which assumes any-to-any connectivity, and hence
%are not general-purpose routing protocols.

Tulip~\cite{abraham05practical} is an overlay network that routes with
$\tilde{O}(\sqrt{n})$ state and round-trip\footnote{Note this is a
different definition of stretch than the more common one-way definition
that we use in this paper.} stretch $2$ on flat names.  Both our work
and Tulip share theoretical techniques based on~\cite{agmnt04}.  However,
like DHTs, Tulip does not solve the network routing problem; its
stretch guarantee effectively assumes a stretch-$1$ routing protocol
beneath it. One could modify Tulip's data structures to store
routes instead of addresses to achieve routing functionality. \ankitr{The problem
then is constructing and maintaining this routing state in a \emph{distributed} and \emph{dynamic}
fashion. Locality aware DHTs~\cite{Abraham04land, 258523} also suffer from this problem. It is this
problem that \disco resolves.}

\parheading{In summary,} to the best of our knowledge, the only proposed
distributed protocol which is both scalable and has low stretch
is~\cite{ford08phd}, and no previously proposed practical protocol
retains guarantees on scalability and stretch while routing on flat
names.

%%%%%%%%%%%%%%%%%%%%%%%%%%%%%%%%%%%%%%%%%%%%%%%%%%%%%%%%%%%%%%%%%%%%%%%%%%%%%%%%

\section{Distributed {C\kern-2pt o\kern-1pt m\kern-2pt p\kern-2pt a\kern-2pt c\kern-1pt t} Routing}
\label{sec:protocol}

%This section begins with an overview of the \disco routing protocol (\S\ref{sec:approach}). 

This section begins with our assumptions and definitions (§\ref{sec:assumptions}). We then present \disco's main components: a {\em name-dependent} distributed compact routing protocol, \nddcr (§\ref{sec:snip}); a name resolution module~(§\ref{sec:nr}); and a distributed location database to achieve {\em name-independence} (§\ref{sec:nip}). Finally, we prove \disco's state and stretch guarantees~(§\ref{sec:guarantees}). We defer an evaluation of messaging overhead to our simulations (§\ref{sec:eval}).

\subsection{Assumptions and definitions}
\label{sec:assumptions}

We assume we are given an undirected connected network of $n$ nodes with
arbitrary structure and link distances (i.e., link latencies or costs).  We let $v
\leadsto w$ denote a shortest path from $v$ to $w$ (in terms of
distance, not hopcount), and let $d(v,w)$ denote the length of this
path. The {\bf name} of a node is an arbitrary bit string; i.e., a flat,
location-independent name.  We assume each node knows its own name and
its neighbors' names, but nothing else. Like past schemes, our protocol
designates some nodes as {\bf landmarks}, whose function will be
described later.  We let $\ell_v$ denote the landmark closest to a node
$v$. 

An event occurs with high probability ({\bf w.h.p.}) if it has
probability $\geq 1 - n^{-c}$ for some constant $c \geq 1$.

We assume nodes can estimate $n$.  While this could be done in many
ways, we propose the use of synopsis diffusion~\cite{nath04synopsis}. SD
requires only extremely lightweight, unstructured gossiping of small
``synopses'' with neighbors and produces robust, accurate estimates
(e.g., within 10\% on average using 256-byte synopses).

\cut{

\parheading{Estimating the number of nodes.} The above protocol assumed
that all nodes know the number of nodes $n$.  Consider first the obvious
simple way to compute the number of nodes: build a tree, and have each
node (beginning with the leafs) report to its parent how many nodes are
in its subtree.  This, however, requires a tree construction protocol
and can be arbitrarily inaccurate in the presence of failures.  On the
other hand, if we add redundancy to the tree, one node might be counted
multiple times. 

Instead, \disco employs synopsis diffusion~\cite{nath04synopsis} which
provides an elegant solution to the problem of robustly estimating $n$. 
SD uses an order- and duplicate-insensitive sketching algorithm of
Flajolet and Martin~\cite{fm85} to produce an estimate of aggregates
like the sum of values across a network.  Our use of SD is for a special
case of this, counting the number of nodes $n$.  SD requires only
extremely lightweight, unstructured gossiping of small ``synopses'' with
neighbors, but results in a robust estimate that can be within an
arbitrarily good constant factor w.h.p. \fixme{is that true?} Although
SD is a one-shot algorithm, it is fairly easy to extend it to a
continuous process. \fixme{describe how}}

%%%%%%%%%%%%%%%%%%%%%%%%%%%%%%%%%%%%%%%%%%%%%%%%%%%%%%%%%%%%%%%%%%%%%%%%%%%%%%%%
\subsection{Name-dependent compact routing}
\label{sec:snip}

\disco begins with a name-dependent distributed compact routing protocol, {\bf \nddcr}. This protocol guarantees worst-case stretch $5$ on the first packet of a flow, worst-case stretch $3$ on subsequent packets, and $\tilde{O}(\sqrt{n})$ routing table entries per node. But it is name-dependent: we assume the source knows the destination's current address (defined below), as opposed to just its name.  \nddcr is based on a centralized algorithm of \cite{tz01}.  In this subsection we describe the components of \nddcr, and then compare it with S4~\cite{mwqls07}, another distributed protocol based on~\cite{tz01}.

\parheading{Landmarks.} A {\bf landmark} is a node to which all nodes know shortest paths. Landmarks will allow us to construct end-to-end routes of the form $s \leadsto \ell \leadsto t$.  Each of the two segments will be precomputed by the routing protocol, and the full route will be close to the shortest if $\ell$ is close to $t$.

Landmarks are selected uniform-randomly by having each node decide locally and independently whether to become a landmark. Specifically, each node picks a random number $p$ uniform in $[0,1]$, and decides to become a landmark if $p < \sqrt{(\log n)/ n}$. Thus, the expected number of landmarks is $n \cdot \sqrt{(\log n)/ n} = \sqrt{n \log n}$ and by a Chernoff bound, there will be $\Theta(\sqrt{n \log n})$ w.h.p.

Since $n$ can change, nodes will dynamically become, or cease to be, landmarks. To minimize churn in the set of landmarks, a node $v$ only flips its landmark status if $n$ has changed by at least a factor $2$ since the last time $v$ changed its status. This amortizes the cost of landmark churn over the cost of a large number ($\Omega(n)$) of node joins or leaves.

\parheading{Vicinities.}  Landmarks gave us a way for a source $s$ to
approximately locate a destination $t$, but if $s$ and $t$ are close, it
will be a poor approximation relative to the distance between them. To
solve this problem, each node $v$ learns shortest paths to every
node in its {\bf vicinity} $V(v)$: the $\Theta(\sqrt{n\log n})$ nodes
closest to $v$.  These sizes ensure that each node has a landmark within
its vicinity w.h.p., which is necessary for the stretch guarantee.

\parheading{Learning paths to landmarks and vicinities.}  Nodes learn
shortest paths to landmarks and vicinities via a single, standard path
vector routing protocol. When learning paths, a route announcement is
accepted into $v$'s routing table if and only if the route's destination
is a landmark or one of the $\Theta(\sqrt{n\log n})$ closest nodes
currently advertised to $v$. The entire routing table is then exported
to $v$'s neighbors.

With these rules, the protocol will converge so that $v$ knows the
landmarks and $V(v)$, for a total of $\Theta(\sqrt{n\log n})$ routing
table entries w.h.p.  We note however that the control plane requires
$\Theta(\delta\sqrt{n\log n})$ state for a node with degree $\delta$,
because path vector stores the full set of route advertisements received
from each neighbor.   This may be acceptable for low-degree graphs or if
routers with higher degree are more well-provisioned. We can reduce
control state to $\Theta(\sqrt{n\log n})$ by either forgetting the
unused announcements as in Forgetful
Routing~\cite{karpilovsky06forgetful}, or by preventing them from being
sent by having each node inform its neighbors of the ``radius'' of its
vicinity.

\parheading{Addresses.} The address of node $v$ is the identifier of its
closest landmark $\ell_v$, paired with the necessary information to
forward along $\ell_v \leadsto v$. Addresses are location-dependent, of course,
but they are only used internally by the protocol, and are dynamically
updated to reflect topology changes.

But what is ``the necessary information to forward along $\ell_v
\leadsto v$''?  In the version of \nddcr evaluated here, it is an {\em explicit
route} consisting of a list of labels, one for each hop along the $\ell_v \leadsto v$ path.
A node's address is thus variable length
with size dependent on the number of hops to its nearest landmark.  The
reader may notice that in the worst case this address size is quite
large, as much as $\tilde{O}(\sqrt{n})$ bits in a ring network of $n$
nodes. This would be too large for a packet header and would explode our
state bound since we will have to store many addresses on some nodes
(§\ref{sec:nr}).  Our design decision to use explicit routes thus invites
some explanation.

Most importantly, the explicit route is in practice extremely compact.
Each hop at a node of degree $d$ is encoded in $O(\log d)$ bits
following the format of~\cite{godfrey09pathlet}. \ankitr{We measured 
the size of explicit routes in CAIDA's
router-level map of the Internet~\cite{caida-router} by picking random
landmarks and encoding shortest paths from each node to its closest
landmark as a sequence of these $O(\log d)$-bit encodings of the node identifiers
on the path.} The maximum size
of our addresses is just $10.625$ bytes (less than an IPv6 address), the
95th percentile is $5$ bytes, and the mean---the important metric for
the per-node state bound---is $2.93$ bytes (less than an IPv4 address).

The explicit route {\em could} be eliminated.  Briefly, an address
would be fixed at $O(\log n)$ bits; each landmark $\ell$ would
dynamically partition this block of addresses among its neighbors in proportion to their number of descendants, and
this would continue recursively down the shortest-path tree rooted at
$\ell$, analogous to a hierarchical assignment of IP addresses.  Since
this would complicate the protocol and actually {\em increase} the mean
address size in practice, we chose the simpler explicit route design.

\parheading{Routing.}  A source $s$ can now send to a destination $t$ as
follows.  If $t$ is a landmark or $t\in V(s)$, then $s$ can route along
a shortest path to $t$.  Otherwise, it extracts $\ell_t$ from $t$'s
address and routes along $s \leadsto \ell_t \leadsto t$.   (Recall that
for \nddcr, unlike \disco, we assume that $s$ knows $t$'s address.) This
first packet of the flow has worst-case stretch $5$, a fact which was
shown in~\cite{tz01} for a centralized algorithm, and which applies to
our protocol since it is essentially a distributed implementation
of~\cite{tz01}.

Subsequent packets can do better. Upon receipt of the message, $t$
determines whether $s\in V(t)$.  If so, $t$ knows the shortest path even
though $s$ didn't (note that $s \in V(t)$ does {\em not} imply $t \in
V(s)$). In this case, $t$ informs $s$ of the path $s\leadsto t$, and all
subsequent packets follow this path.  This is again essentially
equivalent to the protocol of~\cite{tz01} which~\cite{tz01} showed
guarantees worst-case stretch $3$.

\parheading{Shortcutting heuristics.} In S4~\cite{mwqls07}, if at any
point the packet passes through a node which knows a direct path to $t$,
then the direct path is followed. We refer to this as ``To-Destination''
shortcutting. We implement two further optimizations.  First, we try
both the forward and reverse routes $s\to t$ and $t \to s$, and use the
shorter of these. When combined with To-Destination, we call this
``No Path Knowledge'' shortcutting.

Second, a more aggressive optimization is for every node along the route
to inspect the route and see whether it knows a shorter path to any of
the nodes along the route (via its vicinity routes) in either forward or
reverse directions, in which case the route is shortened. We call this ``Up-Down
Stream'' shortcutting. The latter
optimization requires listing the global identifiers of every node along
the path.  This can be done on a single initial packet.  This can also
be combined with using the reverse route (referred to as ``Path Knowledge'' then). 
Our simulations will show that this ``Path Knowledge'' optimization significantly
reduces stretch, but due to the added complexity, we exclude it from the
evaluation of our core protocol. All results discussed subsequently use
the ``No Path Knowledge'' optimization.

\cut{
Table.~\ref{fig:sc-strategies} shows the relative effect of different
shortcutting heuristics.
}
\cut{We implement a further optimization: every
node in the path checks if it has a shorter path to any node downstream
on the path.  Note this is only possible if source routes list global
identifiers rather than a compact representation (see discussion in
Sec.~\ref{sec:assumptions}).}

\parheading{Comparison with S4.}
Both \nddcr and S4~\cite{mwqls07} are distributed protocols based on~\cite{tz01}.  Specifically, S4 is based on the algorithm in Sec.~3 of~\cite{tz01} where, rather than knowing vicinities as described above, each node $v$ knows its ``cluster'': nodes that are closer to $v$ than to their closest landmarks. To adapt~\cite{tz01} to a distributed setting, S4 selects uniform-random landmarks rather than using the multiple-pass algorithm of \cite{tz01}. Unfortunately, this breaks the per-node state guarantee; in \S\ref{sec:eval} we will see that S4's per-node state can be quite high. Briefly, the problem is that some nodes can be close to {\em many} nodes in the network, exploding their cluster size. (This is why~\cite{tz01} needed to select landmarks with a more complex algorithm.)

\nddcr avoids this problem by having each node $v$ store its vicinity (as defined above: the $\tilde{O}(\sqrt{n})$ nodes closest to $v$) rather than its cluster.  This enforces a bound on the number of routing table entries at each node for any network, but does have two consequences.  First, it requires the source to query the destination's vicinity in order to guarantee stretch $3$, as described above; this design is essentially a distributed version of the ``handshaking-based'' scheme of Thorup and Zwick (Sec.~4 of~\cite{tz01}).

Second, this design leads to our use of explicit routes in addresses.  S4 ensures that for any node $v$ and its closest landmark $\ell_v$, node $v$ will appear in $\ell_v$'s cluster. But in \nddcr, it is possible that $v \not\in V(\ell_v)$.  Thus, we include the explicit route $\ell_v \leadsto v$ in $v$'s address, so packets sent to $\ell_v$ can be forwarded to $v$.

\ankitr{S4 and \nddcr do not route on flat names. In the rest of this section, we build \disco, which adds routing on flat names on top of \nddcr.}

\cut{
\fixme{I think we could cut the following paragraph: 
More minor differences from S4 include the use of synopsis diffusion to estimate $n$ (which was not handled in S4) and the use of path vector (rather than two routing protocols SDV and RBDV in~S4).}
}

%%%%%%%%%%%%%%%%%%%%%%%%%%%%%%%%%%%%%%%%%%%%%%%%%%%%%%%%%%%%%%%%%%%%%%%%%%%%%%%%
\subsection{Name resolution}
\label{sec:nr}

\nddcr (§\ref{sec:snip}) requires a sender to know the destination's
current address.  We can solve this by running a consistent
hashing~\cite{karger97consistent} database over the (globally-known) set
of landmarks, similar to~\cite{bvr,mwqls07}.  Every node is aware of its
own address $(\ell_v, \ell_v \leadsto v)$, so it can insert it into the
database, and other nodes can query the database to determine $v$'s
address.   This state is soft: it can be updated, for example, every $t$
minutes and timed out after $2t+1$ minutes.  In our simulator, $t=10$. 
We will show (§\ref{sec:guarantees}) that this adds $O(\sqrt{n \log n})$
entries to each landmark, so the state bound is preserved.

However, this is only a partial solution, because the first packet of a
flow may have arbitrarily high stretch. As discussed in
Sec.~\ref{sec:requirements}, this impacts latency, reliability, and
security. We will, however, use this name-resolution database as a
component of \disco in the next section.

%%%%%%%%%%%%%%%%%%%%%%%%%%%%%%%%%%%%%%%%%%%%%%%%%%%%%%%%%%%%%%%%%%%%%%%%%%%%%%%%
\subsection{Name-independent compact routing}
\label{sec:nip}

\disco is comprised of \nddcr (§\ref{sec:snip}), name resolution
(§\ref{sec:nr}), and a distributed name database.  The last step ensures
constant stretch while routing on flat names, and is what we describe in
this section.

\parheading{Overview.} We build a distributed database which maps nodes'
names to their addresses.  The database requires the key properties that
(1) any necessary query can be accomplished with low stretch, (2) the
$\tilde{O}(\sqrt{n})$ per-node state bound is not violated, and (3)
maintaining the state in the face of network dynamics requires little
messaging. We adopt the idea of ``color''-grouping of nodes from~\cite{abraham2004routing}, but adapted with sloppy grouping and routing schemes which are more amenable to a distributed setting.  We then design a random overlay network organized around these groups which can efficiently distribute flat-name routing state.

\parheading{State.}  We begin with a well-known hash function $h(v)$
(e.g., SHA-2) which maps the node name to a roughly
uniformly-distributed string of $\Theta(\log n)$ bits.  Node $v$ is a
member of a ``sloppy group'' of nodes that have in common the first few
bits of $h(v)$.  Specifically, let $G(v)$ be set of nodes $w$ for which
the first $k:=\lfloor \log_2 (\sqrt{n/\log n}) \rfloor$ bits of $h(w)$
match those of $h(v)$.  Thus, $G(v)$ will contain $O(\sqrt{n \log n})$
nodes w.h.p.  \cut{(Here $c$ is a constant that could be optimized.)}
The group is sloppy because this definition depends on $v$'s estimate of
$n$, which will vary slightly across nodes.  Node $v$ then ensures that
every node in $G(v)$ stores $v$'s address.  We will return to how
exactly this is done shortly.

This definition of the sloppy group $G(v)$ has two important
properties.  First, it is consistent in the sense of consistent
hashing~\cite{karger97consistent}: a small change in the number of nodes
does not result in a large change in the grouping.  This is essential to
scalably handle dynamics.  More precisely, there will be no change in
$k$ and hence no change in the grouping unless $n$ changes by a constant
factor.\footnote{If $n$ is such that $k$ is near the boundary of two
values, we can avoid the frequent ``flapping'' that could result by
changing the sloppy group only when the estimate of $n$ changes by at
least some constant factor (e.g., 10\%).}

The second important property is that if the grouping {\em does} change,
it corresponds to splitting a group in half or merging two groups.  This
means that nodes that have slightly different opinions about the value
of $n$ will still roughly agree on the grouping, which helps us handle
(constant-factor) error in estimating $n$.

We show how to maintain the sloppy group state shortly; but first, we
describe how to use it.

\parheading{Routing.}  To route from $s$ to $t$, node $s$ first checks
(as in \nddcr) whether it knows a direct path to $t$, either because $t$
is a landmark or $t\in V(s)$. If so, it routes directly.

Otherwise, if $s$ knows $t$'s address (i.e., because $s \in G(t)$), it
can route to $t$ according to \nddcr.

Otherwise, $s$ locally computes $h(t)$. It then examines its vicinity
and finds the node $w\in V(s)$ which has the {\em longest prefix match}
between $h(w)$ and $h(t)$. ~\cut{includes the case where s itself
matches $t$'s hash, but do we want to mention this separately?} (This can
be optimized slightly to be the closest node $w$ with a ``long enough''
prefix match.)  Due to our choice of $k$, with
high probability, $w\in G(t)$, so $w$ will know $t$'s current address
$(\ell_t, \ell_t \leadsto t)$. The full path (if no shortcutting occurs)
is thus $s \leadsto w \leadsto \ell_t \leadsto t$.  We show that this
has stretch $\leq 7$ in §\ref{sec:guarantees}.

After the first packet, $s$ knows $t$'s address, so routing proceeds as
in \nddcr with stretch $\leq 3$.

Note that there is a vanishingly small but nonzero probability that
$t$'s address is not found because our state is maintained correctly
only with high probability.  Such events never occurred in our
simulations.  However, if this did occur routing could operate correctly by simply using name resolution on the landmark database (\S\ref{sec:nr}) as a fallback.

\parheading{Sloppy group maintenance.}  We now describe how a node $v$
can ensure that the nodes $G(v)$ have $v$'s address.  (Note that $v$
does not know which nodes these are.) 

A na\"{\i}ve solution is to use the consistent hashing-based name
resolution database running on the landmarks~(\S\ref{sec:snip}).  Each
node $v$ already stores its name and address at the landmark which owns
the key $h(v)$. It is easy to see that all of $G(v)$ will be stored on a
predictable set of $O(\log n)$ landmarks, from which any node can
therefore download its group membership information.  This, however,
imposes a high messaging burden on the landmarks: if every node
changes its address once per minute, the landmark would have to relay
$\tilde{O}(\sqrt{n})$ addresses to each of $\tilde{O}(\sqrt{n})$ nodes
for a total of $\tilde{O}(n)$ bytes per minute.

\disco instead adopts a more decentralized solution. Each node $v$ maintains 
a set of {\bf overlay neighbors} $N(v)$. Similar to a DHT structure, $N(v)$ 
includes $v$'s successor and predecessor in the circular ordering of nodes 
according to their hash values $h(\cdot)$.  $N(v)$ also includes a small 
number of long-distance links called ``fingers''.  

To select a finger, a node $v$ picks a random hash-value $a$ from the \ankitr{part of hash-space that falls within $G(v)$.} Following~\cite{symphony}, $a$ is picked such that the likelihood of picking a value is inversely proportional to its distance in hash-space from $h(v)$. Based on these rules, a node finds the name and address of a finger by querying the landmark-based resolution database for the node with the closest hash-value to $a$. In this manner, node $v$ picks a small constant number of fingers (we will test $1$ and $3$ in our simulations) refreshing the set $N(v)$ periodically at a low rate. It then opens and maintains TCP connections to each of these nodes, for an average of $|N(v)| \approx 4$ or $8$ overlay connections (for $1$ or $3$ fingers respectively) counting both outgoing and incoming connections.

Within this overlay, we can efficiently disseminate routing state in a manner very close to a distance vector (DV) routing protocol. We describe the protocol through its differences from the standard DV protocol. First, we emphasize that we use this protocol only to {\em propagate address information}, rather than to find routes. Second, since we are only interested in disseminating addresses rather than finding short paths, we need not include a distance in the announcement, but we (obviously) must include the originating node's name and address. \ankitr{Third, nodes only propagate advertisements to and from nodes they believe belong to their own group, thus keeping address information local to each group.} Fourth and most importantly, node $v$ propagates advertisements only to those nodes in $N(v) \cap G(v)$ which would cause the message to continue in the same direction: that is, announcements received from an overlay neighbor with higher hash-value are propagated only to neighbors with lower hash-values, and vice-versa. This eliminates distance vector's count-to-infinity problem because as announcements are propagated, their distance (in hash space) from the origin of the announcement strictly increases. Other aspects of the protocol (incremental updates, state maintained, etc.) are similar to DV.

Why does this design work? First note that although nodes differ in their opinion on the sloppy grouping, they don't differ substantially. In particular, since we can ensure that estimates of $n$ are within a factor of $2$ of the correct value w.h.p., nodes will differ by at most {\em one bit} in the number of bits $k$ that they match to determine the grouping. Thus, there will be a ``core group'' $G'(v)$ such that all nodes in $G'(v)$ will agree that they are in the same group. $G'(v)$ is clearly connected via successor/predecessor links, so $v$'s announcement will reach all the nodes in $G'(v)$. Since $|G'(v)|=\Theta(\sqrt{n \log n})$ a Chernoff bound can be used to show that every node $s$ will have some member of $G'(v)$ in its vicinity w.h.p. The routing protocol above will then find such a node with its prefix-matching step, so $s$ will be able to route to $v$.

Beyond {\em correctness}, the design is {\em efficient} in two ways.  First, the overlay has constant average degree, so each node will receive only a few copies of each announcement message.  Second, these messages are propagated relatively quickly, as the expected path length in the Symphony topology is $O(\log^2 n)$\footnote{Since we are gossiping on many paths rather than routing along a greedy path, we conjecture that this bound can be improved to $O(\log n)$ along the lines of~\cite{Manku04knowthy}.}~\cite{symphony}.

We briefly clarify two minor points.  First, there might be conflicting announcements received by $v$ for some node $x$'s address; in this case, $v$ can simply use the announcement which was received from the node whose hash-value is closest to $h(x)$.  Second, during convergence and maintenance of the overlay, $v$ may have all its announcements for some node $x$ temporarily withdrawn even though $x$ is still live. To provide reliable service during these periods, $v$ can delay removal of address state until some short period (e.g., $30$ sec) has passed.

\subsection{Guarantees}
\label{sec:guarantees}

We next prove that \disco maintains low stretch and state. We
do not bound the messaging overhead analytically, but we will simulate
it in §\ref{sec:eval}.  

\parheading{Stretch.} As previously noted, the results of~\cite{tz01}
apply to our name dependent protocol \nddcr: it has stretch $\leq 5$ for
the first packet of a flow, and $\leq 3$ for subsequent packets.

Our name-independent routing, however, lengthens paths further.  We now
show that it still maintains stretch $\leq 7$ for the first packet.

\begin{theorem}
After converging, \disco routes the first packet of each flow with
stretch $\leq 7$, and subsequent packets with stretch $3$ w.h.p.
\end{theorem}

\begin{proof}  Packets after the first have stretch $3$ as shown
in~\cite{tz01}.  For the first packet, there are several special cases
that can result in lower stretch---if the source $s$ knows $t$'s
address, $t\in V(s)$, $t$ is a landmark, or shortcutting occurs; these cases
are easy to deal with and we omit a discussion.
There is also a case that no appropriate $w$ is found and stretch might be higher
than $7$; w.h.p., this
does not occur.
In
the general case, the full path is $s \leadsto w \leadsto \ell_t
\leadsto t$, where $w\in V(s)$, $\ell_t$ is the landmark closest to $t$,
	and each segment is a shortest path. We first show a Useful Fact:
	\begin{eqnarray*}
	d(\ell_t, t) & =    & d(t, \ell_t) \textrm{~~~(graph is undirected)} \\
			    & \leq & d(t, \ell_s)  \textrm{~~~(since $\ell_t$ is $t$'s closest landmark)} \\
			    & \leq & d(t,s) + d(s, \ell_s) \textrm{~~~(triangle inequality)} \\
			    & =    & d(s,t) + d(s, \ell_s)  \textrm{~~~(graph is undirected)} \\
			    & \leq & d(s,t) + d(s,t)  \textrm{~~~($\ell_s$ is in $V(s)$ and $t$ isn't)} \\
			    & =  & 2d(s,t).
			    \end{eqnarray*}
			    Note the next-to-last step ($\ell_s \in V(s)$) holds w.h.p. over the random choice of landmarks.

			    We can now upper-bound the length of each segment of the full path.
			    For the first segment $s \leadsto w$, we have
			    $d(s,w)\leq d(s,t)$ since $w$ is within $s$'s vicinity and $t$ is not.
			    For the second segment $w \leadsto \ell_t$, we have
			    \begin{eqnarray*}
			    d(w,\ell_t) & \leq & d(w,s) + d(s,t) + d(t,\ell_t)  \textrm{~~~(triangle ineq.)} \\
				    & \leq & d(s,t) + d(s,t) + d(t,\ell_t)  \\
				    & \leq & d(s,t) + d(s,t) + 2d(s,t)  \textrm{~~~(Useful Fact)} \\
				    & = & 4d(s,t).
				    \end{eqnarray*} %  \textrm{~~~($w$ is in $V(s)$ and $t$ isn't)}
				    For the third segment $\ell_t \leadsto t$, we have $d(\ell_t, t) \leq
				    2d(s,t)$, again by the Useful Fact. Adding up the three segments, we
				    have that the length of the route is $\leq 7d(s,t)$. \end{proof}

\parheading{State.}  For convenience, we analyze state in terms of the
number of entries in the protocol's routing tables.  Each entry may
contain a node name and address, which in turn contains a landmark name
and an explicit route from the landmark to the destination. (Recall from
\S\ref{sec:assumptions} that the explicit route will in practice occupy
just a few bytes, and for extreme cases a variant of our design can
ensure $O(\log n)$-size addresses for any topology.)

\begin{theorem} After converging, with high probability, each \disco
node of degree $\delta$ maintains $O(\delta \sqrt{n \log n})$ entries in
its routing tables including the control and data planes, or $O(\sqrt{n
\log n})$ using forgetful routing. \end{theorem}

\begin{proof} The protocol maintains several kinds of state. First, each
node of degree $\delta$ stores path vector routing state for the
landmarks. Since each node becomes a landmark independently with
probability $\sqrt{(\log n)/ n}$, there are $\Theta(\sqrt{n \log n})$
landmarks in expectation and, by a Chernoff bound, with high probability
(w.h.p.).   Each of a node's $\delta$ neighbors sends it $O(\sqrt{n \log
n})$ landmark route announcements, for a total of $O(\delta\sqrt{n \log
n})$ entries, or $O(\sqrt{n \log n})$ with forgetful routing
(§\ref{sec:snip}).  Note that even without forgetful routing there
are $O(\sqrt{n \log n})$ entries in the data plane.

Similarly, each node picks the closest $\Theta(\sqrt{n \log n})$ nodes
for its vicinity again via path vector, for a total of $O(\delta\sqrt{n
	\log n})$ entries, or $O(\sqrt{n \log n})$ with forgetful routing.

	To enable compact source routes, each node stores a mapping from a
	compact forwarding label to an outgoing interface.  In general, this
	will require one entry for each of $\delta$ neighbors.  Although it
	would be unrealistic in real topologies, we might have $\delta >
	\sqrt{n}$. However, the node really needs to remember the mapping only
	for those forwarding labels that will actually be used; these will be
		for the neighbors leading along shortest paths to landmarks or nodes in
			the node's vicinity, of which there are at most $O(\sqrt{n \log n})$.

			Landmarks store the name resolution database~(§\ref{sec:nr}) which has
			one entry for each of the $n$ nodes.  These are hashed onto the
			landmarks according to consistent hashing, which in its simplest form
			with a single hash function gives each landmark a factor $\Theta(\log
			n)$ more than their fair share of the keyspace w.h.p.  Thus, the most
			overloaded nodes will receive $n / \Theta(\sqrt{n \log n} / \log n) =
			O(\sqrt{n \log n})$ name resolution entries in expectation, and (by a
			Chernoff bound) w.h.p.  By simply using multiple hash functions we can
			reduce consistent hashing's load imbalance~\cite{karger97consistent} and
			end up with $O(\sqrt{n/ \log n})$ entries.

			The distributed name database (§\ref{sec:nip}) has each node store
			$O(\log n)$ overlay neighbors, and the names and addresses of nodes in
			its sloppy group.  Recall that the group for node $v$ is defined as
			those nodes  which share the first $\lfloor \log_2 (\sqrt{n/\log_2 n}) +
			O(1) \rfloor$ bits of $h(v)$. This includes $\Theta(\sqrt{n \log n})$
			nodes in expectation and, again by a Chernoff bound, with high
			probability.

			Thus, in total the algorithm requires $O(\sqrt{n \log n})$ routing table
			entries.   
			\end{proof}

%%%%%%%%%%%%%%%%%%%%%%%%%%%%%%%%%%%%%%%%%%%%%%%%%%%%%%%%%%%%%%%%%%%%%%%%%%%%%%%%
\section{Evaluation}
\label{sec:eval}

We evaluate \disco in comparison with two recent proposals for scalable
routing, S4 and VRR, over a variety of networks. Our results confirm our
main goals.  First, we find \disco has an extremely balanced
distribution of state across nodes, and across topologies, while in both
S4 and VRR some nodes have a very large amount of routing state in
realistic topologies.  Second, \disco maintains low stretch in all
cases, even for the first packet of a flow, while the other protocols
can have high first-packet stretch; this is particularly evident for
a topology annotated with link latencies rather than simply hopcount.

\subsection{Methodology}

\parheading{Protocols.}  We evaluate five protocols: \disco, \nddcr,
S4~\cite{mwqls07}, VRR~\cite{caesar2006virtual}, and path vector
routing.  \nddcr is our name-dependent protocol coupled with the
landmark-based name resolution database; it is directly comparable to S4
in its goals but guarantees $\tilde{O}(\sqrt{n})$ per-node
state. Our implementation of S4 is as in~\cite{mwqls07} except that we
use path vector for cluster and landmark routing, making it more
comparable to \nddcr. We evaluated VRR with $r=4$ virtual neighbors as
in~\cite{caesar2006virtual}.  VRR's converged state depends on the order
of node joins; we start with a random node and grow the connected
component of joined nodes outward.

\parheading{Simulators.}  We simulated \disco, S4, and path vector in a
custom discrete event simulator.  For topologies larger than $1024$
nodes, we built a static simulator which calculates the post-convergence
state of the network.  We evaluate VRR in the static simulator
only. Our simulator for VRR scales poorly so we present only results
on $1024$ node topologies for VRR.

In many cases, for large topologies, we sample a fraction of nodes or
source-destination pairs to compute state, stretch, and congestion.

\parheading{Topologies.} Our results include
(1) a 30,610-node AS-level map of the Internet~\cite{caida-as}; (2) a
192,244-node router-level map of the Internet~\cite{caida-router};
(3) $G(n,m)$ random graphs of various sizes, i.e., $n$ nodes with $m$
uniform-random edges, with $m$ set so that the average degree is $8$;
and (4) geometric random graphs of various sizes with average degree $8$.

\subsection{Results}

	\begin{figure*}
	\begin{centering}
	\includegraphics[scale=0.76]{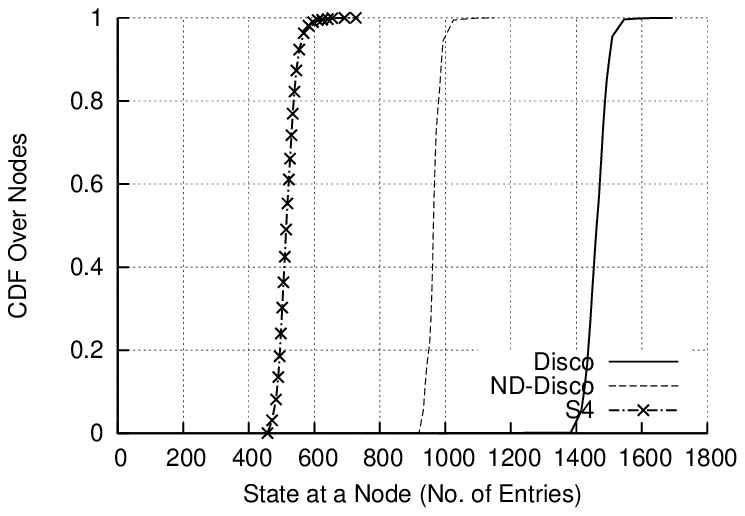}
	\includegraphics[scale=0.76]{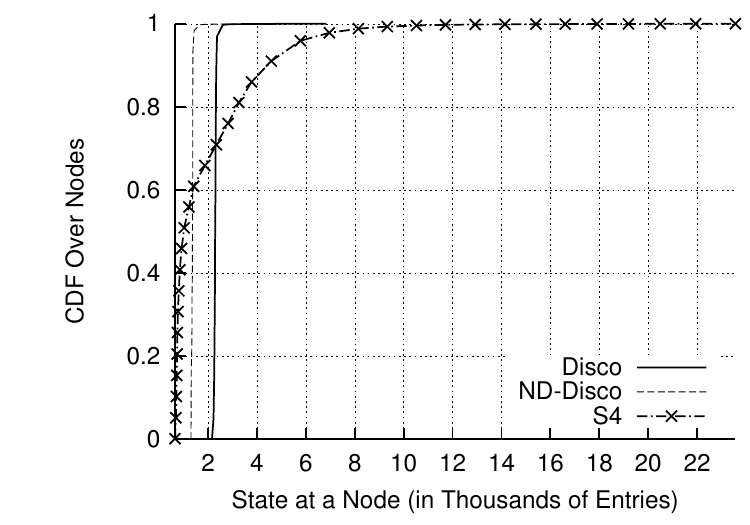}
	\includegraphics[scale=0.76]{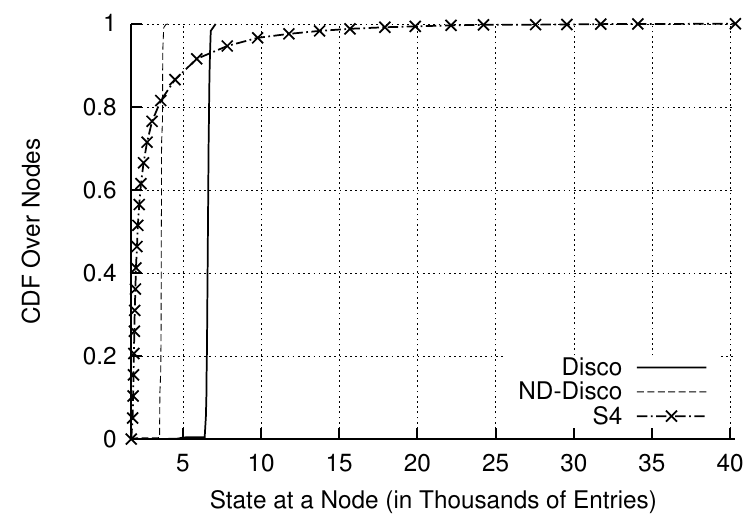}
	\par\end{centering}
	\vspace{-0.1in}
	
	\caption{\label{fig:state} State in a 16,384-node
	Geometric Random Graph (left), 
	Internet AS-level graph (middle) and Internet Router-level graph (right).}
	\end{figure*}

\parheading{State.} We measure data plane state for the protocols.  This
includes everything necessary to forward a packet after the protocol has
converged: forwarding entries for landmarks and vicinities, name
resolution entries on the landmark database, forwarding label mappings
for our compact source route format in \nddcr, and the address
mappings for \disco.

Fig.~\ref{fig:state} shows S4 does well on the random graphs, but is
extremely unbalanced on the Internet topologies. Intuitively, S4 does
poorly on topologies where some nodes are more ``central'' than others. 
In fact, it is easy to show that S4 will have $\tilde{\Theta}(n)$ state
on some nodes in the worst case\footnote{Consider a tree whose root has
$\sqrt{n}$ children at distance $1$; each of these children has
$\sqrt{n}$ children along an edge of distance $2$.  S4's version of
vicinities is called a {\em cluster}: each node $v$ knows the nodes $w$
which are closer to $v$ than the distance from $w$ to its closest
landmark.  Consider any ``grandchild'' node $v$ which is not picked as a
landmark.  The distance to its parent is $2$, but for most such nodes
the parent is not a landmark.  The distance to its grandparent (the
root) is $3$, but with probability $1 - \tilde{O}(1/\sqrt{n})$ the root
is not a landmark.  Therefore the closest landmark must be another
grandchild at distance $4$ from $v$.  Thus, the majority of the
grandchildren nodes must be in the root's cluster, so the root requires
$\Theta(n)$ routing table entries.}. This demonstrates that S4's
simplification of one of the algorithms of~\cite{tz01} can indeed cause
high state.

In contrast, \disco and \nddcr have very balanced distributions of state
in all cases.  Note that \nddcr is a fairer comparison with S4 since
both protocols are name-dependent, while \disco adds additional state
for name-independence.  Average state is slightly higher in \nddcr than
S4 because of a differing design decision about vicinity size: S4
expands its cluster until it reaches a landmark, while \nddcr and \disco
have vicinities which are fixed at $\Theta(\sqrt{n \log n})$ nodes. 
This difference is not fundamental to \nddcr, but is necessary for \disco
so that there will be an intersection between each node's vicinity and
each destination's sloppy group.

\cut{
Note how Disco and ND-Disco have very similar amount of state at all nodes. 
Landmarks do store a marginally larger amount of state than other nodes (small curve
towards the top of the plot) because of the DHT, but this is a requirement for S4 too.
Thus, Disco and ND-Disco always preserve the state bounds, while S4 does not. 
}
	\begin{figure*}
	\begin{centering}
	\includegraphics[scale=0.76]{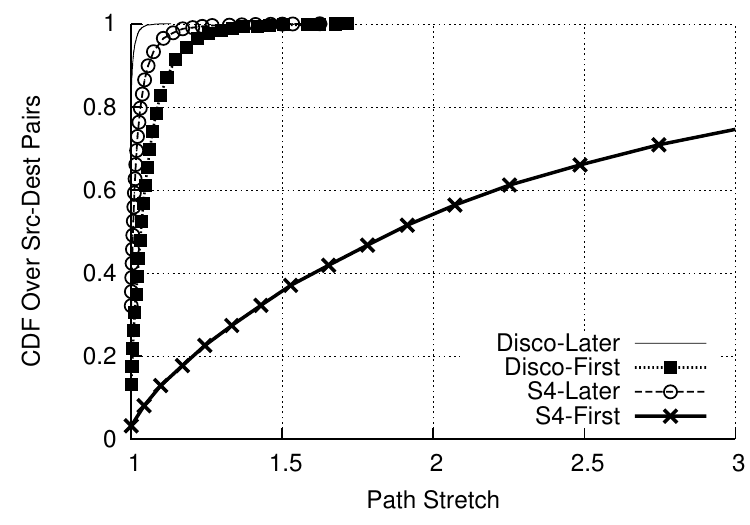}
	\includegraphics[scale=0.76]{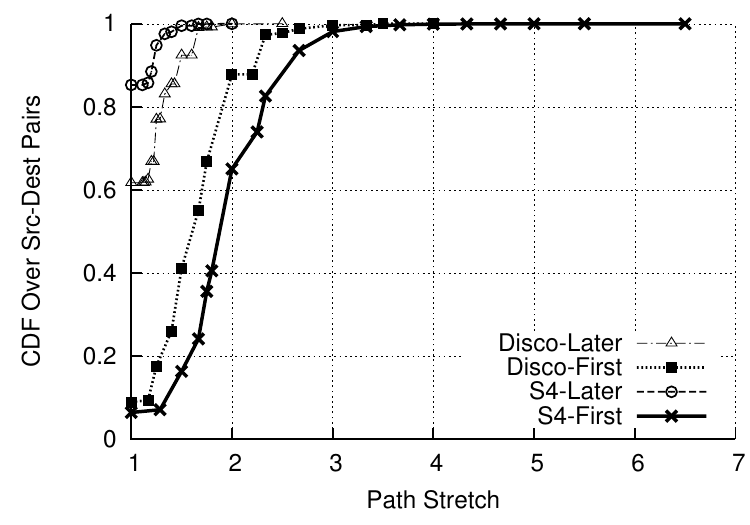}
	\includegraphics[scale=0.76]{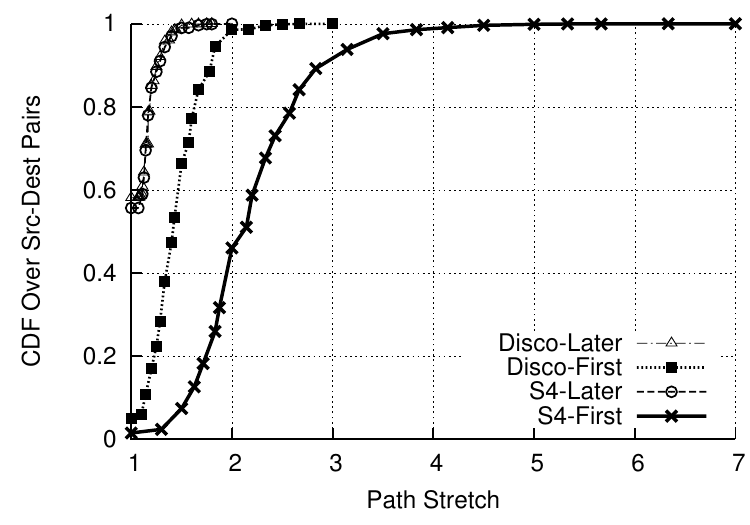}
	\par\end{centering}
	\vspace{-0.1in}
	
	\caption{\label{fig:stretch} Stretch in a 16,384-node Geometric Random Graph (left), 
	Internet AS-level graph (middle) and Internet Router-level graph (right).}
	\end{figure*}

        \begin{figure*}
        \begin{centering}
        \includegraphics[scale=0.76]{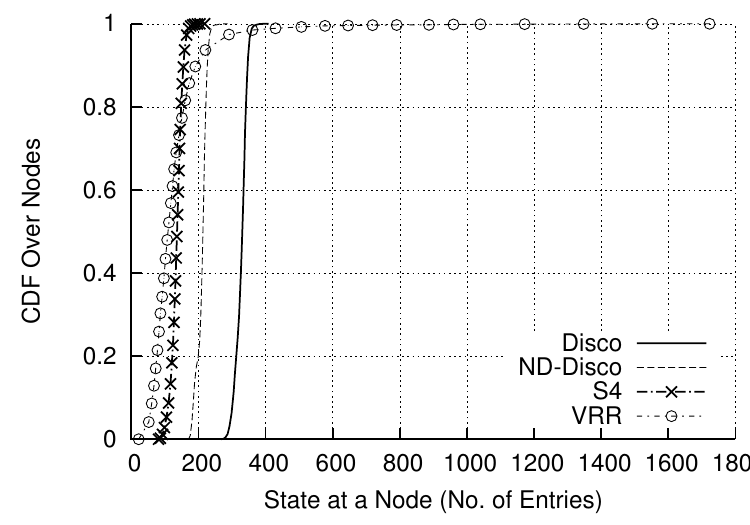}
        \includegraphics[scale=0.76]{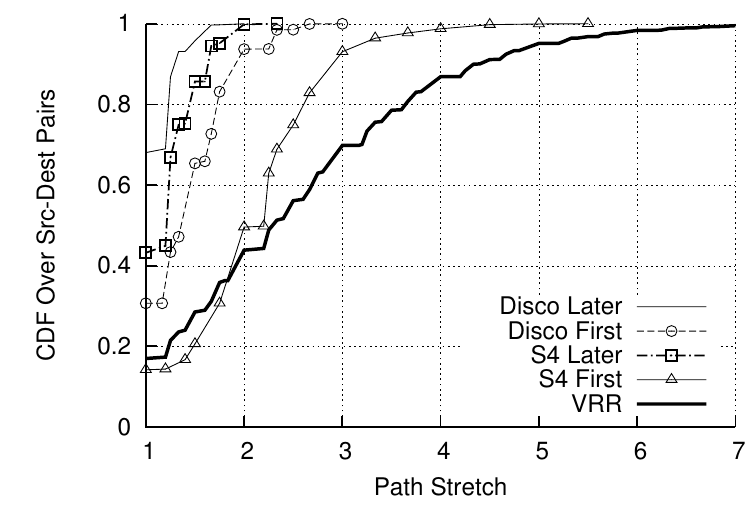}
        \includegraphics[scale=0.76]{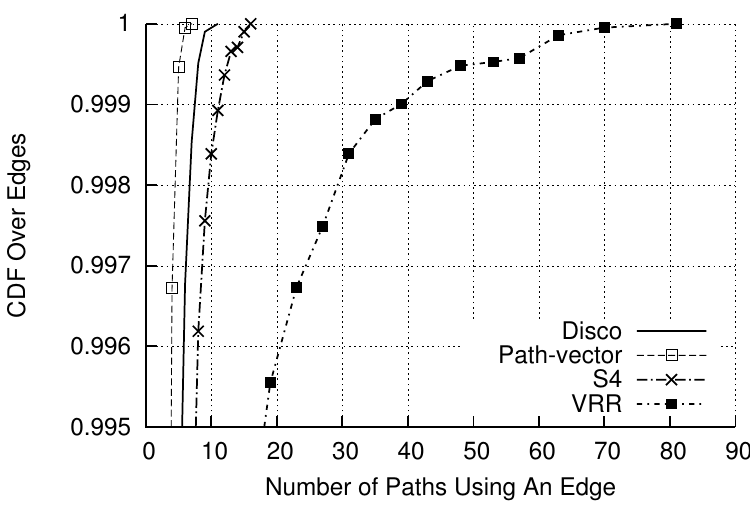}
        \par\end{centering}
        \vspace{-0.1in}
        \caption{\label{fig:random_VRR} State (left), stretch (middle) and congestion (right) comparisons between Disco, VRR and S4 over a 1,024-node $G(n,m)$ random graph.}
        \end{figure*}

        \begin{figure*}
        \begin{centering}
        \includegraphics[scale=0.76]{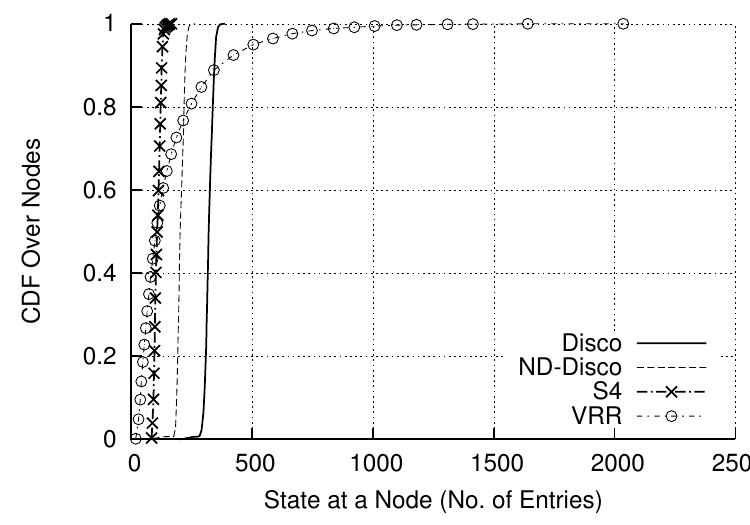}
	\includegraphics[scale=0.76]{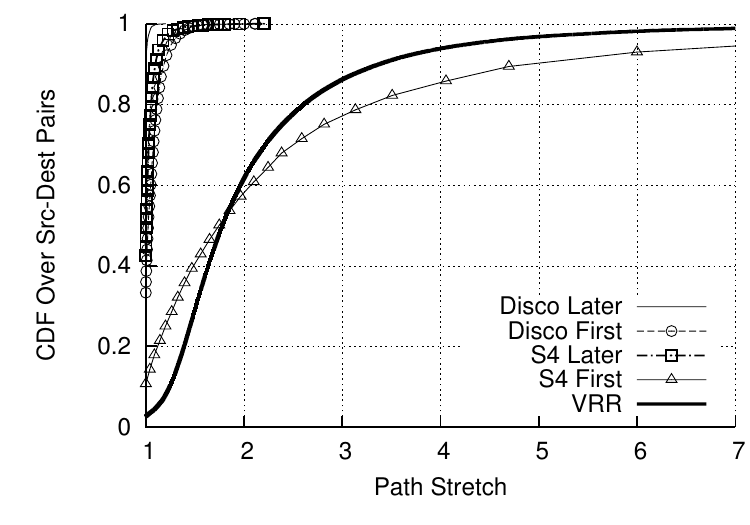}
	\includegraphics[scale=0.76]{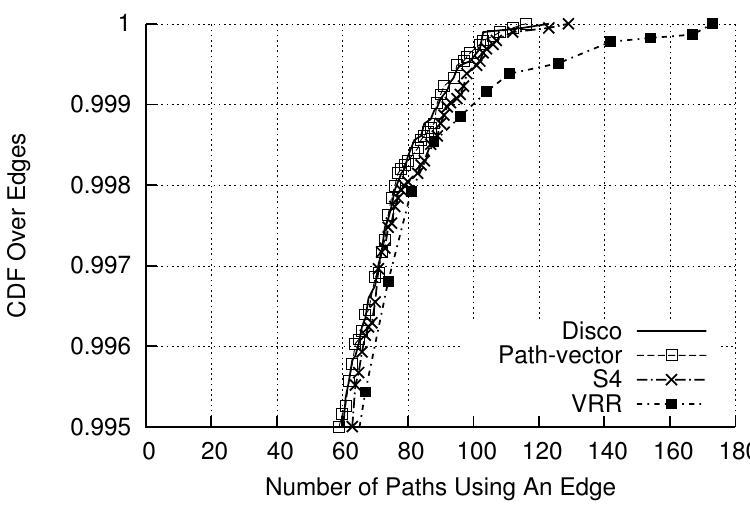}
        \par\end{centering}
        \vspace{-0.1in}
        \caption{\label{fig:geometric_VRR} State (left), stretch (middle) and congestion (right) comparisons between Disco, VRR and S4 over a 1,024-node Geometric random graph.}
        \end{figure*}

	 \begin{figure*}
	         \begin{centering}
		 
		 {\small
		 \begin{tabular}{ l | c | c | c | c }
		          & {\bf AS-Level} & {\bf Router-level} & {\bf Geometric-16384} & {\bf GNM-16384} \tabularnewline
			  \hline
			  \hline
			  No Shortcutting & 1.403 & 1.301  & 1.051 & 1.351 \tabularnewline \hline
			  To-Destination Shortcuts & 1.271  & 1.191 & 1.005 & 1.285 \tabularnewline \hline
			  Shorter\{ReversePath, ForwardPath\} & 1.327 &  1.229 & 1.026 & 1.266 \tabularnewline \hline
			  No Path Knowledge & 1.153  & 1.092 & 1.002 & 1.179 \tabularnewline \hline
			  Up-Down Stream & 1.022  & 1.041 & 1.004 & 1.263 \tabularnewline \hline
			  Using Path Knowledge & 1.007 & 1.015 & 1.002 & 1.159 \tabularnewline \hline
			  \end{tabular}
		 }
		         \par\end{centering}
			         \caption{\label{fig:sc-strategies} Effect of shortcutting strategies: Mean 
				 stretch for different shortcutting heuristics.}
						         \end{figure*}

\begin{figure*}
\begin{centering}
\begin{minipage}{3in}

        \caption{\label{fig:routerlevel_state} Comparison of state at a node for the Router-level 
	Internet topology. S4 does better on average, but severely breaks worst-case bounds.}
\end{minipage}
\hfill
\begin{minipage}{3.8in}
\small{
\begin{tabular}{ l | c | c | c | c | c | c }
	 & \multicolumn{2}{|c|}{{\bf No. of Entries}} & \multicolumn{2}{|c|}{{\bf Bytes (IPv4)}} & \multicolumn{2}{|c}{{\bf Bytes (IPv6)}} \tabularnewline 
\hline
\hline
	& Mean & Max & Mean & Max & Mean & Max \tabularnewline 
\hline
S4 & 3123.9 & 40339.0 & 18.31 & 236.36 & 54.93 & 709.084 \tabularnewline \hline
ND-Disco & 3619.88 & 4310.0 & 21.15 & 34.06 & 63.43 & 100.10 \tabularnewline \hline
Disco & 6592.42 & 7309.0 & 53.03 & 61.608 & 165.31 & 188.20 \tabularnewline \hline
\end{tabular}
}
\end{minipage}
\end{centering}
	\end{figure*}

As Fig.~\ref{fig:random_VRR} (left) and Fig.~\ref{fig:geometric_VRR} (left)
 show, VRR fares very poorly compared to
both S4 and Disco on both topologies in terms of state. It does even
worse than path vector for a few nodes.  This is because VRR constructs
end-to-end paths and stores state at each node along the path, so in theory
it could have as many as $\Theta(n^2)$ routing entries at a node (though it
does not approach this worst case).

All results on routing state discussed so far only count the number of
entries. In Table.~\ref{fig:routerlevel_state}, we present numbers for
state in terms of kilobytes of memory. The size of source routes is
determined using the scheme described in \S\ref{sec:snip}.  As the table
shows, the conclusions are similar when measuring bytes instead of entries.

\parheading{Stretch.} Fig.~\ref{fig:stretch} shows the distribution of
stretch in S4, \disco, and \nddcr.  We call the reader's attention to
the difference between the three graphs.  In the Internet router and AS
topologies, links are unweighted (or equivalently, all link latencies
are $1$).  Thus, maximum stretch is limited simply because the ratio of
the longest to the shortest path is bounded.  In contrast, the geometric
random graph includes link latencies and $S4$ experiences worst-case
stretch of $72$ while \disco's highest stretch is just over $2$.
Unfortunately, a latency-annotated Internet topology was not available
to us; it is likely that this would qualitatively change the protocols'
stretch.

After the first packet, \disco does slightly better than S4 on random
and geometric graphs, both perform similarly on the router-level
topology, and S4 does better on the AS-level graph.

Fig.~\ref{fig:random_VRR} (middle) and Fig.~\ref{fig:geometric_VRR} (middle)
compare stretch in \disco, S4 and VRR over
a $G(n,m)$ random graph and a geometric random graph with latency
values. Like state, VRR provides no bounds on stretch. The maximum
stretch values seen for the first packets in the geometric random
graph are $2.4$ for \disco, $30$ for S4, and $39$ for VRR.
	
	\begin{figure}
	\begin{centering}
	\includegraphics[scale=0.76]{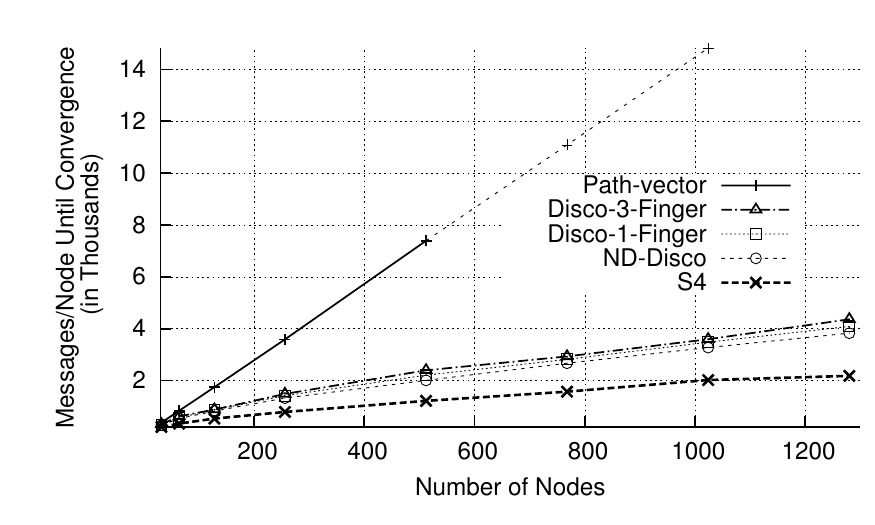}
	\par\end{centering}
	
	\vspace{-0.1in}
	\caption{\label{fig:messaging} Mean messages per node sent until convergence in path vector, S4, \nddcr and \disco (with $1$ and $3$ fingers for address dissemination) for $G(n,m)$ graphs of increasing size. The curve for path vector has been linearly extrapolated beyond 512 nodes.}
	\end{figure}

As described in \S\ref{sec:snip}, the above results use the ``No Path Knowledge''
shortcutting heuristic. Fig.~\ref{fig:sc-strategies} shows the relative effect of different
shortcutting heuristics. Using Path Knowledge, the stretch can be brought
very close to $1$ (last row of table). 
\cut{It is also worth noting how different
heuristics work better for different topologies.}

\parheading{Control messaging.} We compare messaging costs during
initial convergence only, leaving continuous churn to future work.  The
results are shown in Fig.~\ref{fig:messaging}.  \nddcr's messaging
overhead is slightly greater than S4's, reflecting its somewhat larger
vicinities as discussed above.  However, \disco has only a small amount of additional messaging to
support routing on flat names with low stretch, demonstrating the efficiency of our dissemination protocol.  We show results for $1$ and $3$ outgoing fingers per node
in the address dissemination overlay network. The use of a 
larger number of fingers leads to lower overlay diameter,
and hence faster convergence, at the
cost of increased messaging. For instance, for a $1024$-node
$G(n, m)$ topology, with each node picking $1$ outgoing finger, the average and
maximum distances traveled by address announcements were measured to be $5.77$
and $24$ respectively, while picking $3$ random fingers reduced these
numbers to $3.04$ and $16$. At the same time, the number of messages increased by $3.3\%$.

\cut{
Our comparison of messaging costs for convergence is somewhat
rudimentary in that we simulate Disco (and S4) as occurring
phase-by-phase - first nodes find out about their vicinity and the
landmarks, then they store their addresses in the DHT over landmarks,
choose random long-distance neighbors to send addresses of similar-hash
nodes and finally the dissemination of these addresses occurs. We are
working on more accurately modeling this to get better estimates of
messaging costs. The current estimates are summarized in Fig.~\ref{fig:messaging}
}

\parheading{Congestion.} To compute congestion, we have each node route
to a random destination and count the number of times each edge is used.
 One might have guessed that the compact routing schemes
would have high congestion since many routes pass near landmarks. 
However, Fig.~\ref{fig:random_VRR} (right) and
Fig.~\ref{fig:geometric_VRR} (right) show that in synthetic topologies,
congestion is surprisingly close to that of shortest-path routing. VRR
shows higher congestion than all the three other schemes. On the AS-level Internet topology
(Fig.~\ref{fig:asgraph_congestion}), \disco does experience more
congestion for a few edges than shortest-path routing.

\parheading{Scaling.}
Fig.~\ref{fig:scaling} shows how Disco, \nddcr and S4 scale with
increasing number of nodes $n$ in geometric random graphs, showing mean
stretch and mean state. S4's first-packet stretch remains high, but for
the rest of the curves, the stretch is similarly low and close to 1.
Routing state grows as $\tilde{O}(\sqrt{n})$.

	\begin{figure*}
	\begin{minipage}{1.8in}
	\caption{\label{fig:scaling} \disco vs. S4: mean stretch (left) and mean state (right)
	in geometric random graphs of increasing size.}
	\end{minipage}
	\hfill
	\begin{minipage}{5in}
	\includegraphics[scale=0.8]{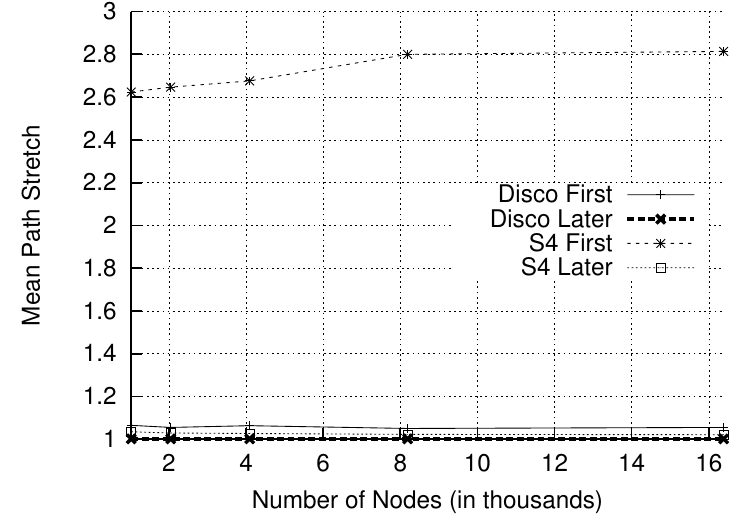}
	\includegraphics[scale=0.8]{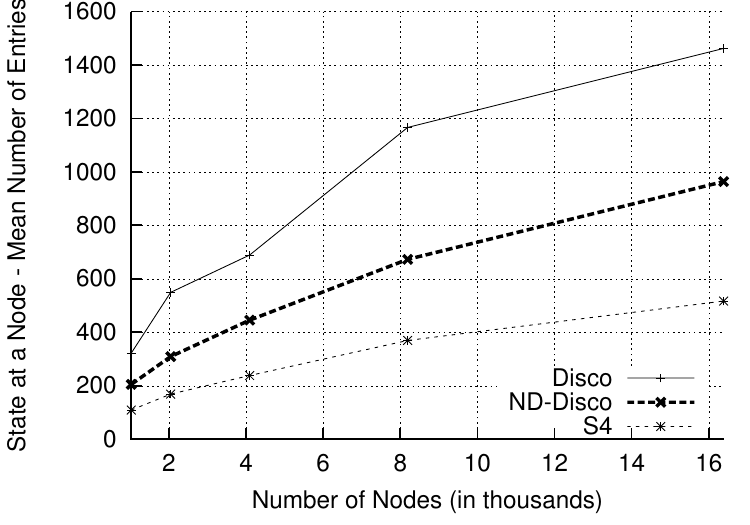}
	\end{minipage}
	\end{figure*}

	\cutatlastminute{

	\begin{figure*}
	\begin{centering}
	\includegraphics[scale=0.8]{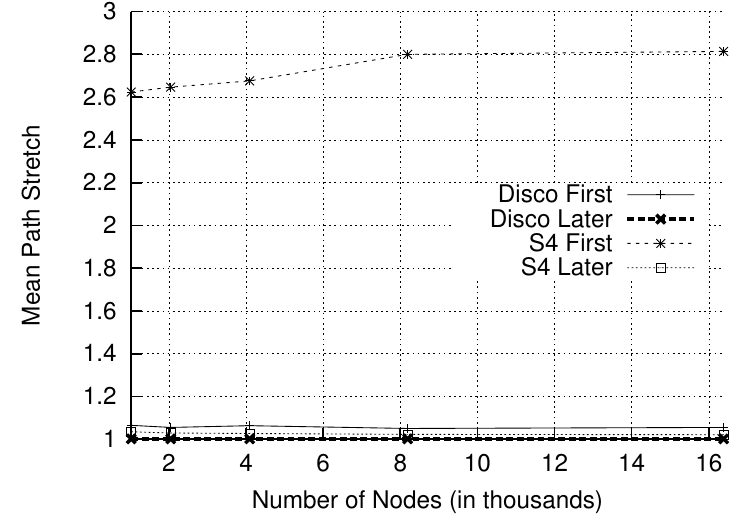}
	\includegraphics[scale=0.8]{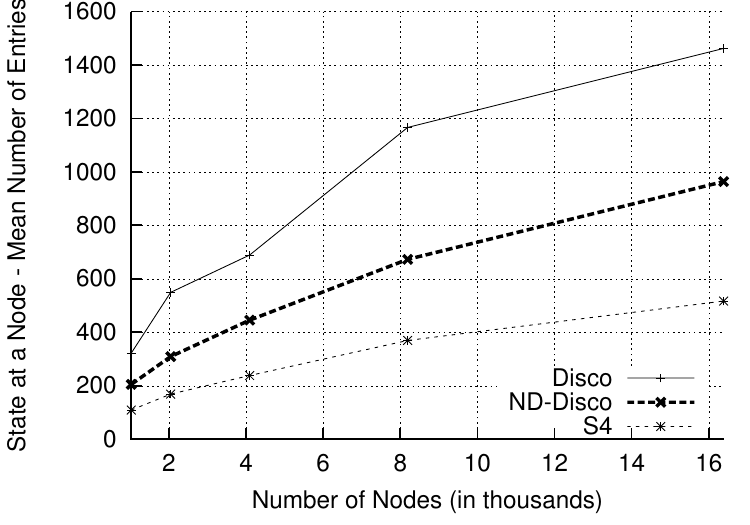}
	\par\end{centering}
	\vspace{-0.1in}
	
	\caption{\label{fig:scaling} \disco vs. S4: mean stretch (left) and mean state (right)
	in geometric random graphs of increasing size.} \end{figure*}
	}

	\begin{figure}
	\begin{centering}
	\includegraphics[scale=0.76]{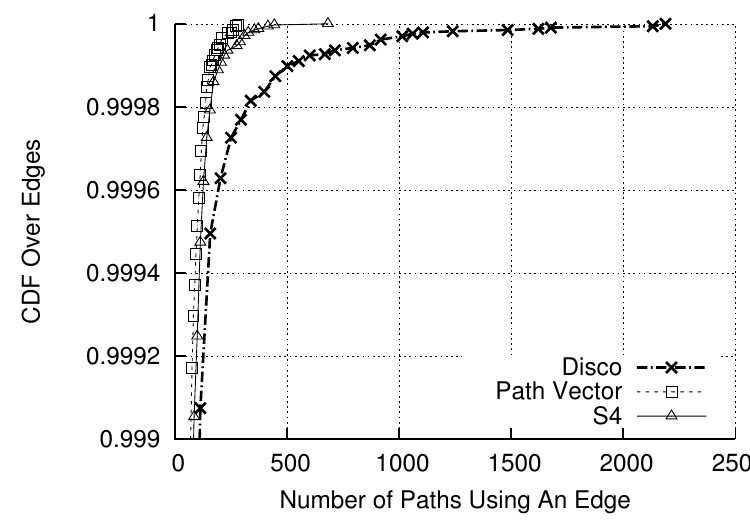}
        \par\end{centering}
        \vspace{-0.1in}

        \caption{\label{fig:asgraph_congestion} On the AS-level Internet topology,
	a small fraction ($0.05\%$) of edges face significantly more congestion than shortest-path routing.}
	\end{figure}

\parheading{Accuracy of static simulation.}
We used a static simulation of the network's state after convergence to
calculate state, stretch and congestion results for large topologies. Our
comparison of results from both the static simulator and the full discrete 
event simulator shows that the static simulator achieves good accuracy.
For instance, for the $1024$-node random graph, the difference between
mean stretch as measured by the static simulator is within $0.9\%$ for
\disco's later packets and $0.7\%$ for S4's later packets. Both stretches
are inflated by similarly small amounts.

\cut{
Fig. shows that results from the static simulator
closely match those from our full discrete event simulator for a
$G(n,m)$ random graph with $1024$ nodes. The only noticeable difference
is in the stretch for \disco's first packet.  This is a result of a
slightly differing implementation: when looking for a destination's
address, the discrete event simulator goes to the node whose
identifier-hash has the longest prefix match with the destination's
identifier-hash. In contrast, the static simulator uses the closest-node
matching the decided number of bits. In all other cases, for any
point on the $x$ axis, the CDF differs by
}
\cut{between the stretch for Disco's first packet - the mean difference
between the two CDF curves from the static evaluation and the simulation
is roughly 0.0657. The fact that there is a non-negligible difference is
unsurprising - when looking for a destination's address, the discrete
event simulator goes to the node whose identifier-hash has the longest
prefix match with the destination's identifier-hash. In contrast, the
static simulator uses the closest-node matching the decided number of
bits. As discussed before, it is possible to optimize to use a
'long-enough match' for getting stretch closer to the static simulator.
All other results are very close; for instance, averaged over the
$x$ axis, the CDF values differ by $0.0037$ for \disco.}

\parheading{Error in Estimating Number of Nodes.}
The previous results assume all nodes know the value of $n$.
Here, we inject random errors of up to 60\% in this estimation.
With $60\%$ random error, across $5$ runs on the $1024$-node random 
graph, only one node failed to find in its vicinity a node
in only one of the sloppy groups, and hence failed to reach all 
destinations in that group. With $40\%$ random error, all nodes were
able to reach all nodes and mean stretch increased marginally by $0.6\%$ from
$1.253$ to $1.261$.  Note that this is an extreme case since
we can ensure error is much lower than $40\%$; §\ref{sec:assumptions}.

\cut{As seen in
Fig.~\ref{fig:numError}, this does not
seem to impact results in a major way---the stretch is similar and no disconnectivity
is observed. The mean difference between the two CDF curves for Disco's later packet
stretch is a very small 0.0034.
}
\cut{
The state does change because each node now has a different notion of
what its vicinity size should be - this results in a larger spread of
the state CDF. In practice, This experiment is only proof of concept -
in practice, we expect Synopsis Diffusion to be significantly more
accurate.
}

%%%%%%%%%%%%%%%%%%%%%%%%%%%%%%%%%%%%%%%%%%%%%%%%%%%%%%%%%%%%%%%%%%%%%%%%%%%%%%%%

\section{Conclusion}
\label{sec:conclusion}

Traditionally, hierarchy has been the only way to scale general-purpose routing. Hierarchy has led to inefficiency in routes, and the use of location-dependent addresses which complicate mobility and management. This paper stands in a long line of work which has progressively brought compact routing, an algorithmically promising approach which eschews hierarchy, closer to practical reality. \disco takes another step forward by providing distributed and dynamic routing on flat names with guaranteed scalability and low stretch.

One area of future work deals with improving our stretch bounds and offering different tradeoffs.  In particular, is it possible to reduce the worst-case first-packet stretch from $7$ to the optimal $3$ in a distributed way?  \disco has chosen one point in the state/stretch tradeoff space, with $\tilde{O}(\sqrt{n})$ state and stretch $\leq 3$ for packets after the first; can we translate other tradeoff points to a distributed setting for name-independent routing?

Another significant outstanding question is to what extent \disco can support policy routing in the Internet.  In many ways, \disco can provide a significant amount of flexibility.  For example, although \disco chooses landmarks randomly, its state and stretch stretch bounds require only that each node has at least one landmark within its vicinity and that there are $\tilde{O}(\sqrt{n})$ total landmarks. These rules would permit an operator to choose landmarks in non-random ways, for example to pick a more well-provisioned landmark, ensure that a node's landmark is within its own domain, or use a landmark service supplied by a network provider.  And to maintain a globally scalable infrastructure with $\tilde{O}(\sqrt{n})$ total landmarks, landmark identifiers could be purchased from or allocated by a registry, much as AS numbers are today. However, policies also pose challenges.  \disco assumes that the route $v \leadsto \ell_v$ can also be used in the reverse direction in $v$'s address in order to route $\ell_v \leadsto v$, and assumes similar reversibility for vicinity routes (when checking the destination's vicinity for a path from the source), which would limit the possible policies. A second problem is that policy routing can significantly lengthen paths; routing through the landmark nearest to a destination many not provide a stretch guarantee for general policies, even when the route length is compared to the shortest policy-compliant path.  Resolving these challenges would be an interesting area of future work.

\cutatlastminute{  This paper brings a fundamentally new approach to the table by eschewing
hierarchy and routing on flat names with guaranteed scalability and
efficiency. \disco opens up several interesting directions for future
research.

First, at a technical level, we leave several open questions. Is it
possible to reduce the worst-case first-packet stretch from $7$ to the
optimal $3$ in a distributed way?   \disco has chosen one point in
the state/stretch tradeoff space, with $\tilde{O}(\sqrt{n})$ state and
stretch $\leq 3$ for packets after the first; can we translate other
tradeoff points to a distributed setting for name-independent routing?

Second, we believe our techniques have the potential to impact multiple
application domains. Content-centric
networks~\cite{jstpbb09,koponen2007data,gc01,sarela2008rtfm} are
particularly relevant since they have extremely large scale in terms of
number of objects, and often use flat self-certifying identifiers.  One
key difference that would have to be overcome is that compact routing
operates at the level of routers, while content-centric networks have a
large number of first-class objects on each physical node.  However, our
techniques may offer a promising starting point for scalable content
routing.}

We thank our shepherd, Bruce Maggs, and the anonymous reviewers for helpful comments.  This work was supported by NSF CNS 10-17069.  The second author was a visiting researcher at Intel Labs Berkeley during part of this project.

%%%%%%%%%%%%%%%%%%%%%%%%%%%%%%%%%%%%%%%%%%%%%%%%%%%%%%%%%%%%%%%%%%%%%%%%%%%%%%%%

\vfill\eject

{\bibliographystyle{abbrv}

\setlength{\itemsep}{-2mm}
\scriptsize{
\bibliography{arxiv}
}}

\end{document}